\def\year{2022}\relax
\documentclass[letterpaper]{article} % DO NOT CHANGE THIS

\newcommand{\ubar}[1]{\underline{#1}}

\newcommand{\prob}{\ensuremath{\operatorname{EPORA}}\xspace}

\newcommand{\tbf}[1]{\textbf{#1}\xspace}
\newcommand{\obja}{\tbf{Max-R}}
\newcommand{\objb}{\tbf{Max-A}}

\newcommand{\alga}{\ensuremath{\operatorname{SAMP}}\xspace}
\newcommand{\algb}{\ensuremath{\operatorname{SAMP-S}}\xspace}

\newcommand{\alggre}{\ensuremath{\operatorname{GREEDY}}\xspace}
\newcommand{\alguniform}{\ensuremath{\operatorname{UNIFORM}}\xspace}
\newcommand{\algrank}{\ensuremath{\operatorname{RANKING}}\xspace}

\usepackage[multiple]{footmisc}
\usepackage[group-separator={,}]{siunitx}

\usepackage{macros_pan}
\usepackage{aaai22}  % DO NOT CHANGE THIS
\usepackage{times}  % DO NOT CHANGE THIS
\usepackage{helvet}  % DO NOT CHANGE THIS
\usepackage{courier}  % DO NOT CHANGE THIS
\usepackage[hyphens]{url}  % DO NOT CHANGE THIS
\usepackage{graphicx}  %Required
\usepackage{graphicx}
\usepackage{natbib}
\usepackage{caption} % DO NOT CHANGE THIS AND DO NOT ADD ANY OPTIONS TO IT
\usepackage{subcaption}
\usepackage{multirow}
\usepackage{xcolor}
\usepackage{xspace}
\usepackage[switch]{lineno}
\DeclareCaptionStyle{ruled}{labelfont=normalfont,labelsep=colon,strut=off} % DO NOT CHANGE THIS
\newcounter{int}
\makeatletter
\newcommand{\citen}[1] {\setcounter{int}{0}\@for\tmp:=#1\do{%
\ifnum \value{int}>0; \fi%
\setcounter{int}{1}%
\citeauthor{\tmp} \shortcite{\tmp}}}
\makeatother
\makeatletter
\newcommand{\citenp}[1]{\setcounter{int}{0}\@for\tmp:=#1\do{%
\ifnum \value{int}>0; \fi%
\setcounter{int}{1}%
\citeauthor{\tmp}, \citeyear{\tmp}}}
\makeatother

\begin{document}
 % \linenumbers

\title{Equity Promotion in Online Resource Allocation\thanks{A preliminary version of this work was presented at the 36th AAAI Conference on Artificial Intelligence. (Corresponding author: Yifan Xu.)}}
\author {
    % Authors
    Pan Xu,\textsuperscript{\rm 1}
    Yifan Xu \textsuperscript{\rm 2}
}
\affiliations {
    % Affiliations
    \textsuperscript{\rm 1} Department of Computer Science, New Jersey Institute of Technology, Newark, USA\\
    \textsuperscript{\rm 2} Key Lab of CNII, MOE, Southeast University, Nanjing, China\\
    pxu@njit.edu,
    xyf@seu.edu.cn
}

\maketitle

%There are lots of news reports showing stark racial disparities existing at least in the early stage of the COVID-19 vaccine rollout across the country. Inspired by this, 

\begin{abstract}
We consider online resource allocation under a typical \emph{non-profit} setting, where limited or even scarce resources are administered by a  not-for-profit  organization like a government. We focus on the internal-equity by assuming that arriving requesters are homogeneous in terms of their external factors like demands but heterogeneous for their internal attributes like demographics. Specifically, we associate each arriving requester with one or several groups based on their demographics (i.e., race, gender, and age), and we aim to design an equitable distributing strategy such that every group of requesters can receive a fair share of resources proportional to \emph{a preset target ratio}.   

We present two LP-based sampling algorithms and investigate them both theoretically (in terms of competitive-ratio analysis) and experimentally based on real COVID-19 vaccination data maintained by the Minnesota Department of Health. Both theoretical and numerical results show that our LP-based sampling strategies can effectively promote equity, especially when the arrival population is disproportionately represented, as observed in the early stage of the COVID-19 vaccine rollout.

\end{abstract}

%\begin{comment}
\section{Introduction}
 We consider online resource allocation under a typical \emph{non-profit} setting, where limited or even scarce resources are administered by a  not-for-profit  organization like a government, and our priority is \emph{equity} such that every type of online agent could receive a fair share of the limited suppliers.  Currently, there are two complementary approaches to studying equity in online resource allocation under the non-profit setting. The first is called \emph{external equity}, which is to treat arriving requesters as homogeneous in terms of their internal attributes like demographics but heterogeneous with respect to their external factors like demands and arrival time. The focus of equity is then placed on fairness about external factors. A few works investigating external equity in distributing medical suppliers during the COVID-19 pandemic~\cite{Manshadi2021FairDR} and allocating food donations to different agencies~\cite{lien2014sequential}. Specifically, these studies associate each arriving request with a demand, and aim to devise an allocation policy such that every requester can get a fair share of resources proportional to its demand, regardless of their arrival time. The second follows an opposite direction, called \emph{internal equity}, which is to treat arriving requesters homogeneous in terms of their external factors like demands but heterogeneous for their internal attributes. Under this approach, each requester has been placed into one or several groups, and a distribution of resources is regarded as equitable when every group of requesters can receive a fair share of resources proportional to the percentage of the group in the whole population.

In this paper, we focus on internal-equity promotion in online resource allocation. A typical motivating example is COVID-19 vaccine distribution. There is much news showing stark racial disparities existing at least in the early stage of the vaccine rollout across the country~\cite{nyt-2,nyt-3,nyt-4,nyt-5,wp-1,wp-2,los-1,stat-1}. The article~\cite{nyt-3} published on Feb.\ 8, 2021, in The New York Times showed that ``In New Jersey, about 48 percent of vaccine recipients were White, and only 3 percent were Black, even though about 15 percent of the state's population is Black, according to state data.'' Another article~\cite{wp-1} published on Feb.\ 1, 2021, in The Washington Post reported that ``White Americans are being vaccinated at rates of up to three times higher than black Americans, though Black Americans have suffered a much higher death rate from COVID-19 than White Americans.'' {These facts highlight the challenge of achieving an equitable distribution of online resources in a non-profit setting, especially when resources are scarce and competed by an overwhelming number of requesters}.

\begin{comment}
The article~\cite{nyt-1} published on March 5, 2021, in The New York Times showed that ``The vaccination rate for Black Americans is half that of White people, and the gap for Hispanic people is even larger, according to a New York Times analysis of state-reported race and ethnicity information.''

~\cite{ma2020group,nanda2020balancing,xu2020trade,Manshadi2021FairDR}\footnote{In~\cite{Manshadi2021FairDR}, authors followed a similar idea to define the equity of rationing as the minimum ratio of the total amount of resources administered to that requested for any arriving user.}
\end{comment}

One technical challenge in addressing (internal) equity promotion in online resource allocation is \emph{how to craft a proper metric of equity}. Current metrics of equity proposed so far are all based on the minimum ratio of the number of requesters served to that of total arrivals within any given group~\cite{ma2020group,nanda2020balancing,xu2020trade}. Unfortunately, these metrics fail to capture the exact equity we expect in practice. Consider COVID-19 vaccine distribution, for example. Among lots of news reporting the sharp racial and ethnic disparities in the early stage of rollout, a commonly cited evidence is the clear gap between the Black/Hispanic share of the vaccinated population and that of a general population or some specific group (like healthcare personnel or long-term care facility residents). This suggests that in an equitable distribution, we expect the \emph{inter-serving-ratio} of each group in some population as close as possible to a specific preset goal. In contrast, current metrics emphasize the \emph{intra-serving-ratio} within each group should be as close as possible to each other. Existing metrics all ignore that we typically have a specific target of serving ratio for each protected group, and the overall equity is often based on how small the gap it is between the serving ratio achieved and the target. In light of the above challenge, we propose a model, called  \emph{Equity Promotion in Online Resource Allocation} (\prob), which features the following three components. 

\begin{comment}
The second challenge is \emph{how to address the dynamic arrivals of demand agents}. Dynamic arrivals of demand agents can be seen in various practical resource-allocation platforms, see, \eg riders in ride-hailing, workers in crowdsourcing markets ([xx]), and buyers in online recommedations [xx].  
\end{comment}

\xhdr{Graph}. There is a bipartite graph $G=(I,J,E)$, where $I$ and $J$ denote the set of types for (offline) supply agents and (online) 
demand agents. Each supply agent (of type) $i$ has a given serving capacity $b_i \in \mathbb{Z}^{+}$, which suggests that agent $i$ can serve up to $b_i$ requesters\footnote{Generally, the serving capacity of each supply agent $i$ can be estimated based on the inventory of resources stored at $i$ and the unit demand per requester.}; each demand agent (of type) $j$ features a certain demographical attributes (\eg race, gender, age, and ethnicity); an edge $e=(i,j)$ indicates the feasibility of supply agent (of type) $i$ to serve demand agent (of type) $j$ due to practical constraints\footnote{Examples includes spacial constraints when each supply and demand agent type features a certain location; see our experimental setup for COVID-19 vaccine distribution in Section~\ref{sec:exp}.}. We have a collection of protected groups $\cG=\{g\}$ (possibly overlapping), where each group $g \subseteq J$ represents a certain class of people. 

\xhdr{Arrival process of demand agents}. We assume each demand agent (of type) $j$ arrives following an independent Poisson process with homegeneous rate $\lam_j>0$, over a time horizon scaled to be $[0,1]$. Upon the arrival of a demand agent $j$, we (as an online algorithm or policy) have to decide immediately and irrevocably either to miss it or assign $j$ to a feasible supply agent $i$ with $(i,j) \in E$ with remaining serving capacity then. 

\xhdr{Equity metrics}. For each group $g \in \cG$, we have a target serving ratio $\mu_g \in (0,1)$ that is known as input.  Consider a given algorithm \ALG and for each demand agent $j$, let $X_j$ be the number of agent $j$ served by $\ALG$ and $A_j$ be that of total arrivals of agent $j$. By assumption, we have $\E[A_j]=\lam_j$. Let $X=\sum_{j \in J}X_j$ and $A=\sum_{j \in J} A_j$ denote the total number of all served demand agents and that of all arrives, respectively. Similarly, for each group $g$, let $X_g =\sum_{j \in g} X_j$ and $A_g =\sum_{j \in g} A_j$ denote the respective number of demand agent served and that of arrivals in group $g$. Consider the below two metrics of equity (which we aim to maximize):\\
 \tbf{Relative-serving ratio} (RSR): $\min_{g \in \cG} \frac{\E[X_g]}{\E[X] \cdot \mu_g}$;
\\
\tbf{Absolute-serving ratio} (ASR): $\min_{g \in \cG} \frac{\E[X_g]}{\E[A] \cdot \mu_g}$.

Note that both RSR and ASR model the gap between the serving ratio achieved the preset target as the relative percentage of the former in the latter. The key difference between the two is how to define the serving ratio achieved. In RSA, the ratio is against the expected number of all agents served ($\E[X]$), while in ASR, it is based on that of all arrivals ($\E[A]$).

\xhdr{Justification of parameters $\{\mu_g\}$}. Observe that
previous studies~\cite{ma2020group,nanda2020balancing,xu2020trade} have chosen a fairness metric as $\min_{g \in \cG} \E[X_g]/\E[A_g]$, the minimum ratio of the expected number of agents served to that of arrivals among all groups. This can be cast as a special case of ASR by setting each $\mu_g=\E[A_g]/\E[A]$. In other words, these previous works assume by default that the preset target serving ratio for each group is equal to its percentage in the arriving population. This  setting of $\{\mu_g\}$ makes sense when each group is perfectly represented in the arriving population. Unfortunately, this doesn't reflect the truth always: disproportionality is indeed observed in practical scenarios like COVID-19 vaccine distribution.~\citen{nyt-1} has reported that Black and Hispanic Americans face much more obstacles to vaccine access compared with their White counterparts and as a result, the arriving population requesting vaccines is dominated by White people (at least in the early rollout stage).

Let \obja and \objb denote the two objectives of maximization of RSR and ASR, respectively. There are several pros and cons associated with each of the above two objectives. Consider a context when  resources are limited or scarce and far outnumbered by online requesters. \tbf{Max-R} ensures only that \emph{the part of resources administered} are distributed according to the preset target ratios among all groups. However, it cannot even guarantee all resources will be used up. In fact, a simple policy $\ALG$ can fool \tbf{Max-R} by selecting a small portion of the arriving requesters to serve such that the served population has a perfect blend of all group members that are well aligned with the preset ratios. In that case, $\ALG$ will achieve an optimal value $1$ on \tbf{Max-R}. In contrast, \tbf{Max-A} will probably lead to an exhaustion of all resources. Additionally, \tbf{Max-A} will direct any policy \ALG to undo the potential bias in the arriving population as follows: for underrepresented groups (\ie $\E[A_g]/\E[A] <\mu_g$), \ALG will try to serve their arriving members as many as possible; for overrepresented groups (\ie $\E[A_g]/\E[A] >\mu_g$), \ALG will instead select only a fraction to serve. That being said, the ASR-based \tbf{Max-A} suffers from the disproportionality in the arriving population, which sets boundaries for the power that an optimal policy can exert on pushing the final ASR toward a preset target.

In light of the above discussions, we will focus on the objective of \tbf{Max-A} in this paper. Note that in our model, we assume $\cI=\{G=(I,J,E), \{b_i\}, \{\lam_j\}, \cG=\{g\}, \{\mu_g\}\}$ are all known as the input. We aim to design an algorithm \ALG such that the \tbf{Absolute-serving ratio} (ASR) is maximized.
% We believe our analyses and results here can be applied there after mild modifications.
%

\xhdr{Other related works}.  There are a few works that have considered online resource allocation in different contexts, see, \eg allocation of public housing for lower-income families~\cite{yair-diversity}, distribution of emergence aid for natural disasters like wildfires~\cite{wang2019measuring}, task assignment in crowdsourcing platform~\cite{DBLP:journals/ton/ChatterjeeBVV18}, online set selection with fairness and diversity constraints~\cite{DBLP:conf/edbt/StoyanovichYJ18}, online resource-allocation problems with limited choices in the long-chain design~\cite{DBLP:journals/mansci/AsadpourWZ20}, and income inequality among rideshare drivers, see, \eg~\cite{NEURIPS2019_3070e6ad,suhr2019two}.

The main model here is called Online Stochastic Matching (OSM) under the arrival setting of Known Identical Independent Distributions (KIID).\footnote{The arrival setting proposed here is slightly different from the standard KIID as studied before but is essentially equivalent; see~\cite{huang2021online}.} OSM was first introduced by~\citet{kvv} and the variant of OSM under KIID  is well-studied before; see \eg~\cite{feldman2009online,haeupler2011online, manshadi2012online, jaillet2013online}. Our model is particularly related to online-side vertex-weighted OSM under KIID where all edges incident to each online agent shares the same weight. In contrast, current works studying vertex-weighted OSM focus on the setting of offline-side vertex-weighted~\cite{ma2021fairness,brubach2016new,huang2018online,aggarwal2011online}, where all edges incident to each offline agent have the same weight. Most of the existing studies simply conducted a worst-case competitive-ratio analysis and gave a constant lower bound. In contrast, we investigate the worst-case more carefully and explicitly examine the different roles played by each external parameter (\eg $\ubar{b}$ and $s^*$) in the final performance.

\subsection{Preliminaries and main contributions}

  \xhdr{Competitive ratio}. The competitive ratio is a commonly-used metric to evaluate the performance of online algorithms. Consider a given algorithm $\ALG$ and a clairvoyant optimal \OPT. Note that $\ALG$ is subject to the real-time decision-making requirement, \ie $\ALG$ has to make an irrevocable decision upon every arrival of online demand agents before observing future arrivals. In contrast, $\OPT$ is exempt from that requirement:~\OPT enjoys the privilege of observing the full arrival sequence of online demand agents \emph{before} optimizing decisions.   Consider an instance $\cI$ of an online maximization problem as studied here. Let $\ALG(\mathcal{I})$ and $\OPT(\mathcal{I})$ denote the expected performance of $\ALG$ and $\OPT$ on $\cI$, respectively. We say $\ALG$ achieves a competitive ratio of $\rho \in [0,1]$ if $\ALG(\cI) \ge \rho \OPT(\cI)$ for all possible instances $\cI$. Essentially, the competitive ratio captures the gap between an online algorithm and a clairvoyant optimal due to the real-time decision-making requirement. 

  \xhdr{Main contributions}.
  Let  $g(s,b):=\E[\min(\Pois(b/s),b)]/b$ with $0 < s \le 1$ and $b \ge 1$, where $\Pois(b/s)$ denotes a Poisson random variable with mean $b/s$.  Here are some properties about $g(s,b)$. We defer the proofs of Lemma~\ref{lem:g} and Theorem~\ref{thm:main-3} to the appendix.

  \begin{lemma}\label{lem:g}
(P1) For any given $0< s \le 1$, $g(s,b)$ is increasing over $b \ge 1$; for any given $b \ge 1$, $g(s,b)$ is decreasing over $s\in (0,1]$. (P2) For any $b \ge 1$, $g(1,b) \ge \max(1-1/\sfe, 1-1/\sqrt{2\pi b})$. (3) For each given $b \ge 1$,  $\lim_{s \rightarrow 0^{+}}g(s,b)=1$.
  \end{lemma}

First, we consider a general case and present an efficient algorithm with a competitive ratio lower bounded by $1-1/\sfe \sim 0.632$.
 % and approaching to $1$ when the minimum serving capacity among all supply agents goes to infinity.
 
 \begin{theorem}\label{thm:main-1}
 There exists an LP-based sampling algorithm (\alga) that achieves an online-competitive ratio of $g(1,\ubar{b}) \ge \max(1-1/\sfe, 1-1/\sqrt{2\pi \ubar{b}})$, where $\ubar{b}=\min_{i\in I} b_i$.  \end{theorem}
 
Second, we consider a special case of homogeneous groups when each group consists of one single type of demand agents. In this context, we refer to types of demand agents and groups interchangeably, and use $\mu_j$ to denote the preset target of serving ratio for each group (or type) $j$. Let $\lam:=\sum_{j \in J} \lam_j$ be the total arrival rates and  $\ubar{\kap}:=\min_{j \in J} \lam_j/(\lam \cdot \mu_j)$, which denotes the degree of  disproportionality for the least represented type in the arriving population with respect to the target serving ratio. Generally, the higher $\ubar{\kap}$ is, the lower disproportionality the arriving population has (\ie every group is well represented in the arriving population). 

\begin{comment}
Generally, the higher $\ubar{\kap}$ is, the lower disproportionality the overall population has. Generally, $\lam_j/(\lam \cdot \mu_j)$ can be interpreted as the ratio of the percentage of type $j$ in the arriving population ($\lam_j/\lam$)  to that we expect to serve ($\mu_j$).
\end{comment}

 \begin{theorem}\label{thm:main-2}
 There exists an LP-based sampling algorithm (\algb) that achieves an online-competitive ratio of $\ubar{\kap} \cdot g(s^*,\ubar{b})$, where $\ubar{\kap}=\min_{j \in J} \lam_j/(\lam \cdot \mu_j)$, $\ubar{b}=\min_{i\in I} b_i$, and $s^*$ is the optimal value to benchmark \LP~\eqref{obj-2}. 
 \end{theorem}
  \noindent\textbf{Remarks}. Let $B=\sum_{i\in I} b_i$ (the total serving capacity among all supply agents), and thus $B/\lam$ (the ratio of the total serving capacities to the total expected number of arrivals) represents the degree of resource scarcity.  We can show that $s^* \le B/\lam$ (See Lemma~\ref{lem:mu}) and therefore, in practical cases when resources are highly scarce, we have $s^* \le B/\lam  \ll 1$ and $g(s^*,\ubar{b}) \sim 1$ by Lemma~\ref{lem:g}. As such, \algb will outperform \alga when $\ubar{\kap}$ is relatively large (\ie the arriving population has a relatively low disproportionality overall).

  \begin{theorem}\label{thm:main-3}
The online-competitive analyses for \alga and \algb are both tight (\ie the ratios presented above are the best ones they can ever achieve). Also, no algorithm can achieve an online-competitive ratio better than $\sqrt{3}-1 \sim 0.732$ even for homogeneous groups.
   \end{theorem}
   
We test our model and algorithms on publicly available COVID-19 vaccination datasets maintained by the Minnesota Department of Health. Experimental results confirm our theoretical predictions and
demonstrate the power of our policies in navigating the distribution of limited resources toward the preset target ratios when compared against heuristics; see more details in Section~\ref{sec:exp}.

% \vspace{-0.1in}
\section{An LP-based policy for the general case}
Throughout this paper, we use \OPT to denote both of a clairvoyant optimal policy and its corresponding performance. For each pair of supply-demand agents $e=(i,j)$, let $x_{ij}$ denote the expected number of times that $j$ is assigned to be served by $i$ in a clairvoyant optimal (unconditional of the number of arrivals of $j$).   Let $\cN_i$ ($\cN_j$) be the set of neighbors of agents with respect to $i$ ($j$) in the input graph $G=(I,J,E)$. Let $\lam=\sum_{j \in J} \lam_j$. Consider the below linear program (\LP). 
% \begingroup
% \allowdisplaybreaks
\begin{alignat}{2}
\max & ~~ s  &&  \label{obj-1} \\
 &x_j :=\sum_{i \in \cN_j} x_{ij} \le \lam_j  &&~~ \forall  j\in J   \label{cons:j-1} \\ 
  & x_g:=\sum_{j \in g}\sum_{i \in \cN_j} x_{ij} \ge s \cdot \lam \cdot \mu_g&& ~~ \forall  g\in \cG \label{cons:g-1} \\
 & x_i:=\sum_{j \in \cN_i} x_{ij} \le b_i  && ~~ \forall  i \in I \label{cons:i-1}\\
 & s, x_{ij} \ge 0 && ~~\forall (ij)\in E\label{cons:e}
\end{alignat}
% \endgroup
%$\OPT \le \min(1,s^*)$, where \OPT denotes the expected performance of a clairvoyant optimal and $s^*$ is the optimal value to \LP~\eqref{obj-1}.
\begin{lemma}\label{lem:lp-1}
$\OPT \le s^*$, where \OPT denotes the expected performance of a clairvoyant optimal and $s^*$ is the optimal value to \LP~\eqref{obj-1}.
\end{lemma}
\begin{proof}
Recall that our objective is \tbf{Max-A}, maximization of the absolute-serving ratio (ASR). 
Observe that the expected number of total served demand agents in each group $g$ is equal to $\E[X_g]=x_g$ while the expected number of total arrivals is $\E[A]=\lam$. Thus, by the definition of ASR, our objective can be expressed as $\max \min_{g \in \cG} x_g/(\lam \cdot \mu_g)$. By setting  $s= \min_{g \in \cG} x_g/(\lam \cdot \mu_g)$, we see Constraint~\eqref{cons:g-1} ensures $\max s$ is equivalent to the \tbf{Max-A}. In the rest, we try to justify all other constraints are valid for any clairvoyant optimal policies. Constraint~\eqref{cons:j-1} is valid due to the expected number of served agents $j$ (LHS) should be no larger than that of total arrives (RHS); Constraint~\eqref{cons:i-1} is valid since the expected number of agents served by $i$ (LHS) should be no larger than its capacity (RHS); the last constraint is straighforward. Thus, we claim that \LP~\eqref{obj-1} is maximizing the expected performance on ASR over all possible clairvoyant optimal policies.
\end{proof}

 For the ease of notation, we use $\{x_{ij}\}$ to denote an optimal solution of \LP~\eqref{obj-1} when the context is clear. Based on  $\{x_{ij}\}$, we propose such a simple algorithm (\alga) as stated in~Algorithm~\ref{alg:a}, that is to non-adaptively sample a neighbor $i \in \cN_j$ with probability $x_{ij}/\lam_j$ upon every arrival of $j$ and assign it whenever $i$ has remaining capacity. Observe that the sampling distribution is valid since $\sum_{i \in \cN_j} x_{ij}/\lam_j \le 1$ due to Constraint~\eqref{cons:j-1}.
 
\begin{proof}[Proof of the main Theorem~\ref{thm:main-1}]
 Consider a given supply agent $i$, we see that each time its capacity will be consumed at a rate of $ \sum_{j \in \cN_i} \lam_j \cdot (x_{ij} /\lam_j)=\sum_{j \in \cN_i} x_{ij} \le b_i$. This suggests that at any time $t$, the probability that $i$ has at least one remaining capacity should be at least $\Pr[\Pois(b_i t)<b_i]$. Let $X_{ij}$ be the (random) number of times that demand agent $j$ is served by $i$ in \alga. We have
 {
\begin{align}
\E[X_{ij}] &\ge \int_0^1 \lam_j \cdot \frac{x_{ij}}{\lam_j}\cdot \Pr[\Pois(b_i t)<b_i] dt\\
&=\frac{x_{ij} }{b_i} \int_0^1  \Pr[\Pois(b_i t)<b_i] \cdot b_i \cdot dt \label{ineq:pa}\\
&=x_{ij}\cdot  \frac{\E[\min(\Pois(b_i),b_i)]}{b_i}.
\end{align}
}
Observe that the integral part on the second line above \eqref{ineq:pa} essentially counts the number of arrivals in a Poisson process with an arrival rate $b_i$ whenever the total number of arrivals so far is less than $b_i$ and thus, it is equal to $\E[\min (\Pois(b_i),b_i)]$. For each group $g \in \cG$, let $X_g=\sum_{j \in g} \sum_{i \in \cN_j} X_{ij}$ denote the number of agents in group $g$ served in \alga.
{
\begin{align}
&\frac{\E[X_g]}{\E[A] \cdot \mu_g \cdot s^*} =\frac{\sum_{j \in g}\sum_{i \in \cN_j}\E[X_{ij}]}{\lam \cdot  \mu_g \cdot s^*}\\
& \ge \frac{\sum_{j \in g}\sum_{i \in \cN_j}x_{ij}}{\lam \cdot  \mu_g \cdot s^*}\cdot  \frac{\E[\min(\Pois(b_i),b_i)]}{b_i} \\
& \ge \frac{\E[\min(\Pois(b_i),b_i)]}{b_i} = g(1,b_i) \ge g(1,\ubar{b}). \label{ineq:a}
\end{align}
}
As for Line~\eqref{ineq:a}: the first inequality is due to Constraint~\eqref{cons:g-1}; the second inequality is due to (P1) in Lemma \ref{lem:g} (the function $g(s,b)$ is increasing on $b$) and $b_i \ge \ubar{b}$. Thus,  we have $ \min_{g \in \cG} \E[X_g]/(\E[A] \cdot \mu_g) \ge g(1,\ubar{b}) \cdot s^*$. By Lemma~\ref{lem:lp-1}, we establish the competitive ratio of \alga.
\end{proof}

\begin{comment}
Thus, we claim that $\OPT \le s^* \le 1$. Consider a natural LP-based sampling algorithm. That is to sample a neighbor $i \in \cN_j$ with probability $x_{ij}^*/(\lam \cdot \mu_j \cdot s^*)$. Note that perhaps we should remove Constraint~\eqref{cons:j-1} later, in that case $\OPT=\min(s^*,1)$.

Observe that for each given $i$, its capacity will be consumed at a rate of 
\[
\sum_{j \in \cN_i} \lam_j \cdot \frac{x_{ij}^*}{\lam \cdot \mu_j \cdot s^*} \le \frac{\kap_U \cdot b_i}{s^*},
\]
where $\kap:=\max_j  \lam_j/(\lam \cdot \mu_j)$, a parameter capturing the largest degree of disproportionality in the arrival population. Note that for each given $(ij)$, it will be served with a chance
\begin{align*}
\E[Z_{ij}]&=\lam_j \int_o^1 \frac{x_{ij}^*}{\lam \cdot \mu_j \cdot s^*}\ \Pr\Big[ \Pois(\kap \cdot b_i \cdot t/s^*)<b_i \Big]dt \\
&=\frac{\lam_j \cdot x_{ij^*}}{\lam \cdot \mu_j} \frac{1}{\kap b_i } \int_0^1 \frac{\kap b_i}{s^*} \Pr\Big[\Pois(\kap b_i t/s^*) <b_i \Big]dt\\
&=\frac{\lam_j \cdot x_{ij^*}}{\lam \cdot \mu_j} \frac{\E\Big[\min\Big(b_i, \Pois(\kap b_i/s^*)\Big) \Big]}{\kap b_i}
 \end{align*}
\end{comment}

\begin{algorithm}[t!]
\DontPrintSemicolon
\textbf{Offline Phase}: \;
%\tcc{The offline phase will take as input $(G=(I,J,E), \cG=\{g\},\{b_i| i\in [m]\}, \{\lam_j, \mu_j| j\in [n]\})$.}
Solve  \LP~\eqref{obj-1}, and let $\{x_{ij}\}$  be an optimal solution.\;
\textbf{Online Phase}:\;
Let a demand agent (of type) $j$ arrive.\;
{Sample a neighbor $i \in \cN_j$ with probability $x_{ij}/\lam_j$.}\;
\eIf{$i$ has remaining serving capacity,}{Assign $j$ to $i$;}{Reject $j$.}
\caption{An LP-based sampling  (\alga).}
\label{alg:a}
\end{algorithm}

 \section{A special case: homogenous groups}
Recall that for homogeneous groups, each group consists of one single type of demand agents, and $\mu_j$ denotes the preset target serving ratio for each type (or group) $j$. Generally, $\{\mu_j\}$ reflect the demographics of all types among a certain population (\eg the general population, or some special groups such as infected people). For this reason, we assume that $\mu:=\sum_{j \in J} \mu_j \ge 1$ (note that in case there are overlaps among types, the sum can exceed one). Similar to before, let $x_{ij}$ be the expected number of times that $j$ is served by $i$ in a clairvoyant optimal (\OPT). Observe that for homogeneous groups, we have each $g=\{j\}$ and we can verify that LP-\eqref{obj-1} then is reduced to the below.

% $\mu:=\sum_{j \in J} \mu_j=1$.\footnote{WLOG assume there are no duplicate types, and thus $\mu \le 1$. We can make it equal by adding some dummy types if necessary.} Similar to before, let $x_{ij]$

\begin{alignat}{2}
\max & ~~ s  &&  \label{obj-2} \\
 &x_j :=\sum_{i \in \cN_j} x_{ij} \le \lam_j  &&~~ \forall  j\in J   \label{cons:j-2a} \\ 
  & x_j:=\sum_{i \in \cN_j} x_{ij} \ge s \cdot \lam \cdot \mu_j && ~~ \forall  j\in J\label{cons:j-2b} \\
 & x_i:=\sum_{j \in \cN_i} x_{ij} \le b_i  && ~~ \forall  i \in I \label{cons:i-2}\\
 & s, x_{ij} \ge 0 && ~~\forall (ij)\in E\label{cons:e-2}
\end{alignat}
Parellel to Lemma~\ref{lem:lp-1}, we have the below claim.
\begin{lemma}\label{lem:lp-2}
The optimal value of~\LP~\eqref{obj-2} serves as a valid upper bound for the expected performance of a clairvoyant optimal. 
\end{lemma}

We present another LP-based sampling algorithm as follows. For the ease of notation, we use $\{x_{ij},s^*\}$ to denote an optimal solution to \LP~\eqref{obj-2}, and \emph{assume WLOG that $x_j=\sum_{j\in J}x_{ij}=s^* \cdot \lam \cdot \mu_j$ for each $j \in J$}. Note that we refer to agents $j$ with $\kap_j:=\lam_j/(\lam \cdot \mu_j)>1$ and $\kap_j<1$ as overrepresented and underrepresented, respectively.

%For each overpresented agent $j$ with $\kap_j>1$, we sample a neighbor $i \in \cN_j$ with probability $x_{ij}/(\lam_j \cdot s^*)$; for each underpresented agent $j$ with $\kap_j<1$, we sample a neighbor $i \in \cN_j$ with probability $x_{ij}/(\lam \cdot\mu_j \cdot s^*)$.  
  
 \begin{algorithm}[t!]
\DontPrintSemicolon

\textbf{Offline Phase}: 
Solve  \LP~\eqref{obj-2}, and let $\{x_{ij}, s^*\}$ be an optimal solution with each $x_j=\sum_{i\in \cN_j}x_{ij}=s^*\cdot\lam \cdot \mu_j $.\;

\textbf{Online Phase}:\;
Let a demand agent (of type) $j$ arrive.\;
\eIf{$j$ is an overpresented agent with $\kap_j:=\lam_j/(\lam \cdot \mu_j) >1$}
{Sample a neighbor $i \in \cN_j$ with probability $x_{ij}/(\lam_j \cdot s^*)$. Assign $j$ to $i$ if $i$ still can serve.}{Sample a neighbor $i \in \cN_j$ with probability $x_{ij}/(\lam \cdot \mu_j \cdot s^*)$.  Assign $j$ to $i$ if $i$ still can serve.}
\caption{An improved policy for homogeneous groups (\algb).}
\label{alg:b}
\end{algorithm}

Recall that $\ubar{\kap}=\min_{j\in J}\kap_j$, which denotes the degree of disproportionality of
 the least represented type in the arriving population, and $B=\sum_{i \in I} b_i$.  The following observations will be useful later.
 %based on the condition $\mu=\sum_{j \in J} \mu_j \ge 1$. 
 
\begin{lemma}\label{lem:mu}
(P1) $s^* \le \ubar{\kap} \le 1$; (P2) $s^* \le B/\lam$.
\end{lemma}

% \begin{comment}%proof of lemma 4
\begin{proof}
We show the first inequality in (P1). Observe that in the optimal solution, we assume WLOG  that $x_j=s^*\cdot \lam \cdot \mu_j$ for all $j \in J$. By Constraint~\eqref{cons:j-2a}, $x_j \le \lam_j$, and thus $s^*\cdot \lam \cdot \mu_j \le \lam_j$, which suggests that $s^* \le \lam_j/(\lam \cdot \mu_j)=\kap_j$ for each $j$. Therefore, $s^*\le \min_{j\in J} \kap_j=\ubar{\kap}$. We show the second one in (P1). By definition of $\ubar{\kap}$, we have $\lam_j \ge \lam \cdot \mu_j \cdot \ubar{\kap}$ for all $j \in J$. Thus, $\sum_{j \in J} \lam_j=\lam \ge  \sum_{j\in J} \lam \cdot \mu_j \cdot \ubar{\kap}=\lam \cdot \mu \cdot \ubar{\kap}$, which suggests that $\ubar{\kap} \le 1/\mu \le 1$ since $\mu \ge 1$.

Now we show (P2). Observe that in the optimal solution,
$$\sum_{j \in J} x_j=\sum_{j \in J}s^*\cdot \lam \cdot \mu_j =s^*\cdot \lam \cdot \mu=\sum_{i \in I} x_i \le \sum_{i \in I} {b_i}=B,$$
which implies that $s^* \le B/(\lam \cdot \mu) \le B/\lam$.
\end{proof}
% \end{comment}

\begin{comment}
The first inequality $\OPT \le s^*$ follows from the same analysis in Lemma~\ref{lem:lp-1}. We focus on the second one $s^*\le 1$. Observe that
{\small
\[s \cdot \lam \cdot \mu=\sum_{j \in J} s \cdot \lam \cdot \mu_j \le \sum_{j\in J}x_j 
 \le \sum_{j \in J} \lam_j =\lam,\]
 }
where the first inequality is due to Constraint~\eqref{cons:j-2b} and the second one is due to Constraint~\eqref{cons:j-2a}. Thus, we claim that $s \cdot \mu \le 1$, which suggests that $s^* \le 1$ since $\mu \ge 1$.
\end{comment}

\begin{lemma}\label{lem:S}
Sampling distributions in \algb are valid for all demand agents and they can be viewed as a boosted version of those in \alga.
\end{lemma}

% \begin{comment}%proof of lemma 5
\begin{proof}
We first show the validity. For each overrepresented agent $j$ with $\kap_j=\lam_j/(\lam \cdot \mu_j)>1$, we have 
% {\small
% \begin{equation}
\begin{align*}
\sum_{i \in \cN_j}x_{ij}/(\lam_j \cdot s^*) & =x_j/(\lam_j \cdot s^*) =(s^* \cdot \lam \cdot \mu_j)/(\lam_j \cdot s^*) \\ &=(\lam \cdot \mu_j)/\lam_j <1.
\end{align*}
% \end{equation}
% }
 For each underrepresented agent $j$ with $\kap_j \le 1$, 
% {\small
\[
\sum_{i \in \cN_j} x_{ij}/(\lam \cdot \mu_j \cdot s^*)=x_j/(\lam \cdot \mu_j \cdot s^*)=(s^* \cdot \lam \cdot \mu_j)/(\lam \cdot \mu_j \cdot s^*)=1.\]
% }
For each overrepresented agent $j$, we have $x_{ij}/(\lam_j \cdot s^*) \ge x_{ij}/\lam_j$ since $s^* \le 1$. For each underrepresented agent $j$, note that $\lam \cdot \mu_j \cdot s^* = x_j  \le \lam_j $  for all $j \in J$. This suggests that $x_{ij}/(\lam \cdot \mu_j \cdot s^*) \ge x_{ij}/\lam_j$. Thus, we claim that sampling distribution for each demand agent in \algb can be treated as boosted compared with that in \alga.
\end{proof}
% \end{comment}

\begin{proof}[Proof of the main Theorem~\ref{thm:main-2}]
For each supply agent $i$, let $\cN_{i,H}$  ($\cN_{i,L}$) denote the set of neighbors of $i$ that are overrepresented (underrepresented). Observe that for each $i$, its capacity will be consumed in \algb with a rate of 
{
\[
\sum_{j \in \cN_{i,H}} \lam_j \cdot \frac{x_{ij}}{\lam_j \cdot s^*}+\sum_{j \in \cN_{i,L} }\lam_j \cdot \frac{x_{ij}}{\lam \cdot \mu_j \cdot s^*}   \le \frac{b_i}{s^*}.
\]
}
\tbf{Case 1}. Consider a given underrepresented type $j$, let $X_{ij}$ denote the (random) number of times that $j$ is served by $i$ in \algb. 
\begin{align*}
\E[X_{ij}] &=\int_0^1 \lam_j \cdot \frac{x_{ij}}{\lam \cdot \mu_j \cdot s^*}\cdot \Pr[\Pois(b_i t/s^*)<b_i] dt\\
&=\frac{\kap_j \cdot x_{ij} }{b_i} \int_0^1  \Pr[\Pois(b_i t/s^*)<b_i] \cdot \frac{b_i}{s^*} \cdot dt\\
&=\kap_j \cdot x_{ij}\cdot  \frac{\E[\min(\Pois(b_i/s^*),b_i)]}{b_i}\\
&=\kap_j \cdot x_{ij}\cdot g(s^*,b_i) \ge \ubar{\kap} \cdot x_{ij}\cdot g(s^*,\ubar{b})
\end{align*}
The last inequality is due to (P1) in Lemma \ref{lem:g} (the function $g(s,b)$ is increasing on $b$) and $b_i \ge \ubar{b}$. 

\tbf{Case 2}. Consider a given overrepresented type $j$ and let $X_{ij}$ be the number of times that $j$ is served by $i$.
\begin{align*}
\E[X_{ij}] &=\int_0^1 \lam_j \cdot \frac{x_{ij}}{\lam_j \cdot s^*}\cdot \Pr[\Pois(b_i t/s^*)<b_i] dt\\
&=\frac{x_{ij}}{b_i} \int_0^1  \Pr[\Pois(b_i t/s^*)<b_i] \cdot \frac{b_i}{s^*} \cdot dt\\
&=x_{ij}\cdot  \frac{\E[\min(\Pois(b_i/s^*),b_i)]}{b_i}\\
&= x_{ij}\cdot g(s^*,b_i) \ge  \ubar{\kap} \cdot x_{ij}\cdot g(s^*,\ubar{b})
\end{align*}
The last inequality is valid since $\ubar{\kap} \le 1$. Thus, summarizing the above two cases, we claim for any pair of supply-demand agent $(ij) \in E$, we have $\E[X_{ij}] \ge \ubar{\kap} \cdot x_{ij}\cdot g(s^*,\ubar{b})$. Let $X_j=\sum_{i \in \cN_j} X_{ij}$ denote the number of agents of $j$ served in \algb. 
\begin{align*}
&\frac{\E[X_j]}{\E[A]\cdot \mu_j \cdot s^*} =\frac{\E[X_j]}{\lam \cdot \mu_j \cdot s^*}
=\frac{\sum_{i\in \cN_i} \E[X_{ij}]}{x_j}\\
&\ge
\frac{\sum_{i \in \cN_j}\ubar{\kap} \cdot x_{ij}\cdot g(s^*,\ubar{b})}{x_j}=\ubar{\kap} \cdot g(s^*,\ubar{b}), \end{align*}
which suggests that\\ $\min_{j \in J} \E[X_j]/(\E[A] \cdot \mu_j) \ge (\ubar{\kap} \cdot g(s^*,\ubar{b}))\cdot s^*$. By Lemma~\ref{lem:lp-2}, we establish the competitive ratio of \algb.
\end{proof}

\begin{figure*}[t!]
    \centering
    \includegraphics[scale=0.75]{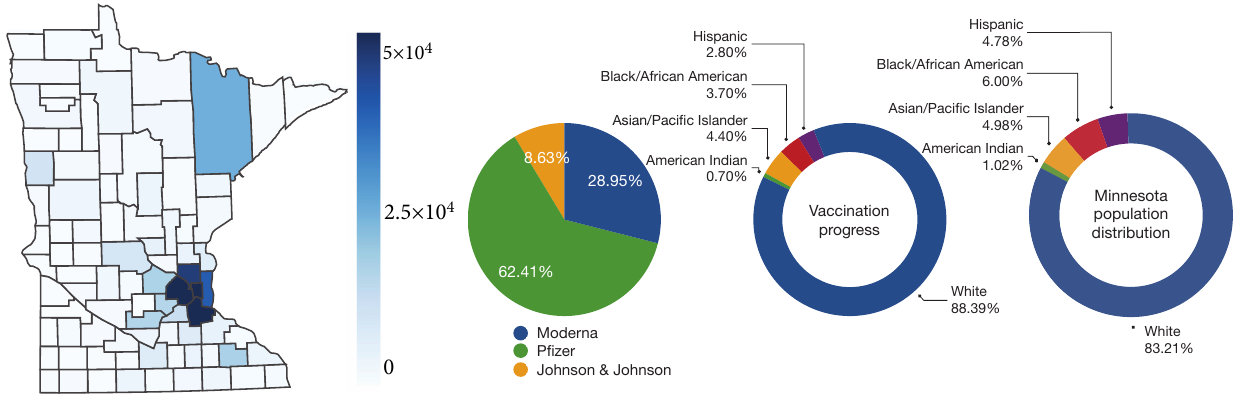}
    % \vspace{-2mm}
    \caption{\small Overview of the COVID-19 vaccination datasets from May 7 to June 17, 2021, in Minnesota: the distribution on the numbers of people fully vaccinated over counties (first), the percentage of the three vaccines administered (second), the percentage of people fully vaccinated by race (third), and the distribution of the general population by race (fourth), respectively.}
    \label{fig:overview}
\end{figure*}

\section{Experimental results}\label{sec:exp}
Due to the space limit, we present only the experiemental results on the general case. We also conduct experiments on the special case of homogeneous groups and defer the results to the appendix.

\subsection{Experiments on the general case}
We use the publicly available COVID-19 vaccination datasets that are maintained by the Minnesota Department of Health\footnote{\url{https://mn.gov/covid19/vaccine/data/index.jsp}.}. 

\begin{comment}
%provided by Minnesota Immunization Information Connection (MIIC)

% \begin{figure}[h!]
%   \begin{minipage}[t]{0.45\linewidth}
%     \centering
%     \includegraphics[scale=0.5]{resource_dist.pdf}
%     \caption{Percentage of vaccines distributed.}
%     \label{fig:resource_dist}
%   \end{minipage}%
%   \hspace{5mm}
%   \begin{minipage}[t]{0.45\linewidth}
%     \centering
%     \includegraphics[scale=0.5]{race_dist.pdf}
%     \caption{Minnesota population distribution.}
%     \label{fig:race_dist}
%   \end{minipage}
% \end{figure}
\end{comment}
% \vspace{-4mm}

\begin{figure*}[ht!]
  % \vspace{-2mm}
  \begin{minipage}[t]{0.32\linewidth}
    \centering
    \includegraphics[scale=0.5]{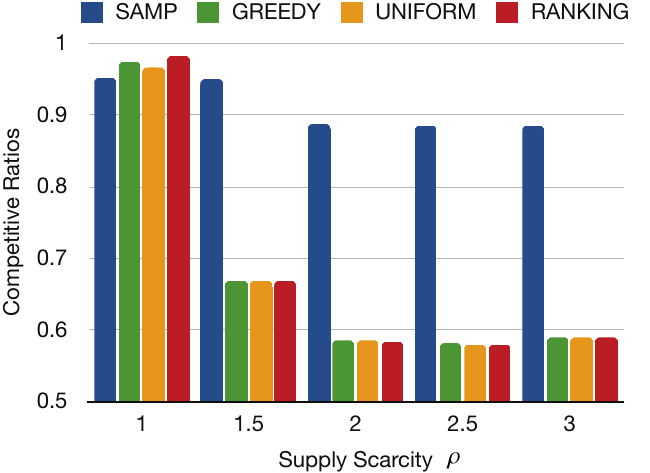}
    % \vspace{-3mm}
    \caption{\small Competitive ratios achieved when the supply scarcity $\rho\in\{1,1.5,2,2.5,3\}$.}
    \label{fig:cr_scar}
  \end{minipage}%
  \hspace{3mm}
  \begin{minipage}[t]{0.32\linewidth}
    \centering
    \includegraphics[scale=0.5]{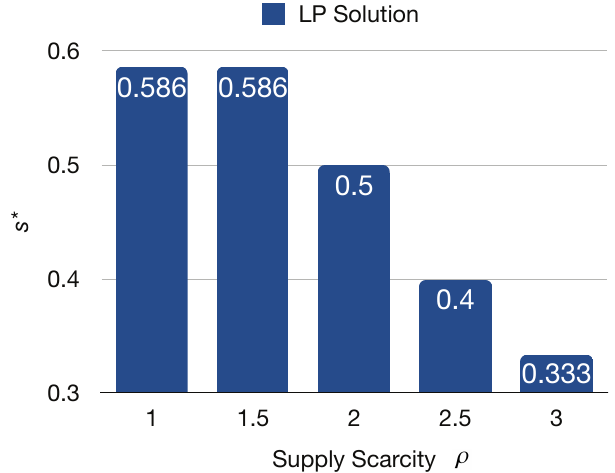}
    % \vspace{-3mm}
    \caption{\small The optimal value $s^*$ of \LP~\eqref{obj-1} as the supply scarcity $\rho\in\{1,1.5,2,2.5,3\}$.}
    \label{fig:opt_scar}
  \end{minipage}
  \hspace{3mm}
  \begin{minipage}[t]{0.32\linewidth}
    \centering
    \includegraphics[scale=0.5]{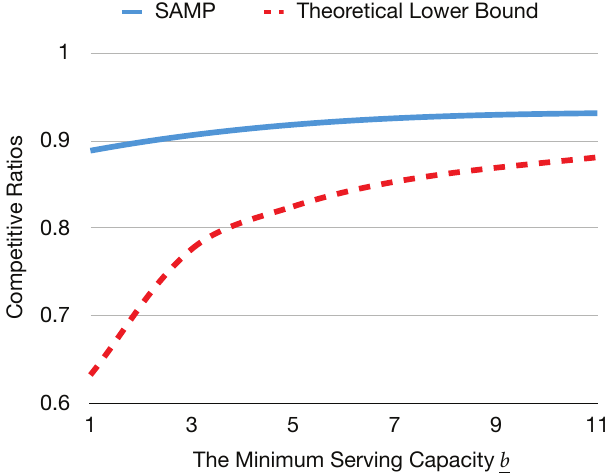}
    % \vspace{-3mm}
    \caption{\small Competitive ratios achieved when the min serving capacity $\ubar{b} \in \{1,3,5,7,9,11\}$.}
    \label{fig:cr_b}
  \end{minipage}
  % \vspace{-4mm}
\end{figure*}

% \bluee{In Figure~\ref{fig:race_dist}, we should fully spell the race names.}
%``People Vaccinated, By County'' that shows
%\begin{itemize}
% ``Doses Shipped to Minnesota Providers, by Product'' (see Figure~\ref{fig:resource_dist})
% (see Figure~\ref{fig:race_dist})

\xhdr{Preprocessing of the COVID-19 vaccination datasets}. The datasets include the following information relevant to our experimental setup: (1) The total cumulative numbers of people  with completed vaccine series in any given county (updated weekly); (2) the cumulative numbers of doses of each vaccine product (Pfizer, Moderna, and Johnson \& Johnson) shipped to to local Minnesota providers and pharmacies in the C.D.C. partnership program;  (3) the weekly progress regarding the percentages of people with completed vaccine series in any age-race group, where age is among 15+, 15-44, 45-64, and 65+, and race has options of American Indian (AI), Asian/Pacific Islander (API), Black/African American (BAA), Hispanic (H), and White (W). We extract the data collected from (1), (2), and (3) over the time range from May 7 to June 17 in 2021, and plot them in Figure~\ref{fig:overview}. The first three figures in Figure~\ref{fig:overview} (from left to right) show the distribution on the numbers of people that get fully vaccinated over different counties, the percentage of the three vaccines administered, and the percentage of people getting fully vaccinated by race, respectively. We also map the distribution of the general population in Minnesota by race based on American Community Survey (ACS) 5-year estimates. 

Based on the above information, we construct the input instance as follows. Figure~\ref{fig:overview} shows that Minnesota is made up of $87$ counties. For each pair of county and vaccine type, we create a supply agent. We set the \emph{initial} serving capacity for  each supply agent being equal to the total number of doses of the specific vaccine administered in the county divided by the number of units needed per requester. To make it computationally tractable, we downsample each $b_i$ proportionally such that the 
total serving capacity $B=\sum_{i \in I} b_i=10^4$. For each county-race pair, we create a demand agent, and set the \emph{initial} arriving rate being equal to the percentage of the people getting fully vaccinated in that specific county and race. Motivated by the work~\cite{Manshadi2021FairDR}, we introduce a concept, called \emph{supply scarcity} (denoted by $\rho \ge 1$), which represents the ratio of the expected total of demand to that of  supply. We inflate the arriving rate $\lam_j$ of each demand agent $j$ by a factor of $B \times \rho$ such that the total arriving rates after scaling up will satisfy $\lam=\sum_{j \in J} \lambda_j=B \times \rho$. We add an edge between a demand agent and a supply agent if the two agents locate in the same or adjacent counties (the relative positives of all the $87$ counties are shown in Figure~\ref{fig:overview}). We create a group for each race, and set the preset goal  as the percentage in the general population as illustrated in the last plot of Figure~\ref{fig:overview}. The statistics of the baseline instance is summarized in Table~\ref{tab:sta}. 

\begin{table}[t!]
\scriptsize
\caption{\small Statistics of the Baseline Instance.}
\centering
\begin{tabular}{c c}
\hline \hline
Parameters & Settings \\ [0.5ex]
\hline
$I,J,E$ & $|I| = 87\cdot 3=261$, $|J| = 87\cdot 5=435$, $|E| = 8483$ \\ 
\hline
Serving capacities $\{b_i\}$ & Min : $1$, Max : $1822$, Mean : $38.3$ \\ \hline
Arriving rates $\{\lam_j\}$ & Min : $0.01$, Max : $2580$, Mean : $23$ \\ \hline
%Groups $\cG=\{g \}$ & {\footnotesize{American Indian (AI), Asian/Pacific Islander (API), Black/African American (BAA), Hispanic (H), and White (W)}}\\ \hline
Preset goals $\{\mu_g\}$ & AI:$1.02\%$,API:$4.98\%$,BAA:$6\%$,H:$4.78\%$,W:$83.21\%$ \\
\hline
\end{tabular}
\label{tab:sta}
% \vspace{-4mm}
\end{table}

%, and thus, there are in total $|I|=87 \cdot 3=261$ supply agents. 

%We focus on the vaccination data from May 7, 2021 to June 17, 2021 with completed vaccine series and assume that all available doses are offline supply agents while vaccine recipients are online demand agents that arrive dynamically. Note that Minnesota is made up of $87$ counties that are well defined and do not overlap, which can be seen in in Figure~\ref{fig:people_vaccinated}.  Thus we can categorize all supply agents according to the pre-defined $87$ counties and the types of vaccine product. In this way, each supply agent $i$ is featured by its location (county) and provider (Pfizer, Moderna or Johnson \& Johnson). In addition, we set the initial inventory $b_i$ for each supply agent $i$ according to (\textbf{D1}) and (\textbf{D2}).
% In the original dataset (\textbf{D2}), there are $352,075$ administered doses in total.On the other hand, each demand agent $j$ is featured by its location (county) and race (AI, API, BAA, H or W\footnote{Here we exclude Multilateral (M) since it has less representativeness when compared with other races.}).We can re-construct the arriving rate $\lambda_j$ for each $j$ from (\textbf{D2}) and (\textbf{D3}) within the given time range.Note that an edge $e = (i,j)$ indicates that $i$ and $j$ are at the same county or adjacent to each other. For each (race) group $g$, we set the preset goal $\mu_g$ as the percentage in the general population from (\textbf{D4}).

%Next, we describe how we construct the total supplies and demands.

\begin{figure*}[ht!]
  \centering
  \begin{subfigure}[b]{0.31\linewidth}
    \includegraphics[width=\linewidth]{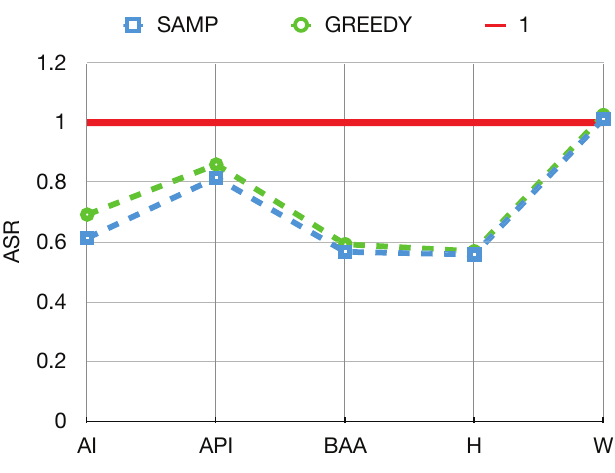}
    % \vspace{-2mm}
    \caption{ASR achieved when $\rho = 1$.}
    \label{fig:asr_s1}
  \end{subfigure}
  \begin{subfigure}[b]{0.31\linewidth}
    \includegraphics[width=\linewidth]{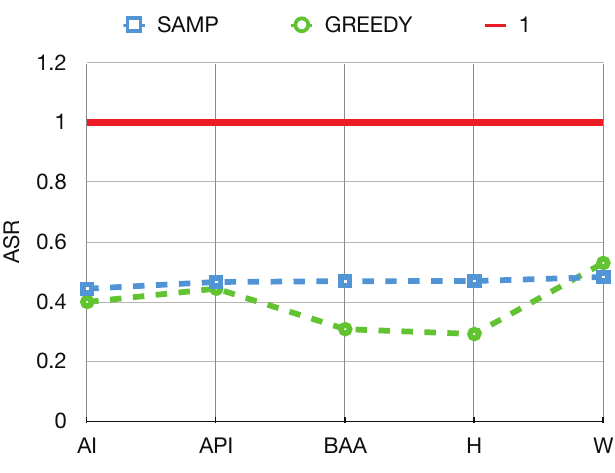}
    % \vspace{-2mm}
    \caption{ASR achieved when $\rho = 2$.}
    \label{fig:asr_s2}
  \end{subfigure}
  \begin{subfigure}[b]{0.31\linewidth}
    \includegraphics[width=\linewidth]{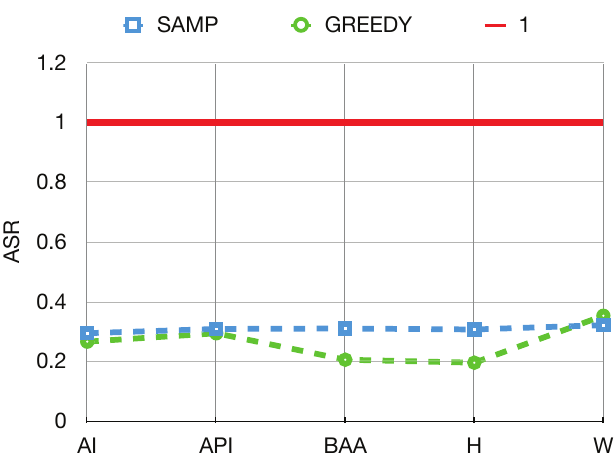}
    % \vspace{-2mm}
    \caption{ASR achieved when $\rho = 3$.}
    \label{fig:asr_s3}
  \end{subfigure}
  \vspace{-2mm}
  \caption{ASR achieved for different races, \ie American Indian (AI), Asian/Pacific Islander (API), Black/African American (BAA), Hispanic (H), and White (W)), when the supply scarcity $\rho \in \{1,2,3\}$.}
  \label{fig:asr}
\end{figure*} 

\begin{figure*}[ht!]
  \centering
  \begin{subfigure}[b]{0.31\linewidth}
    \includegraphics[width=\linewidth]{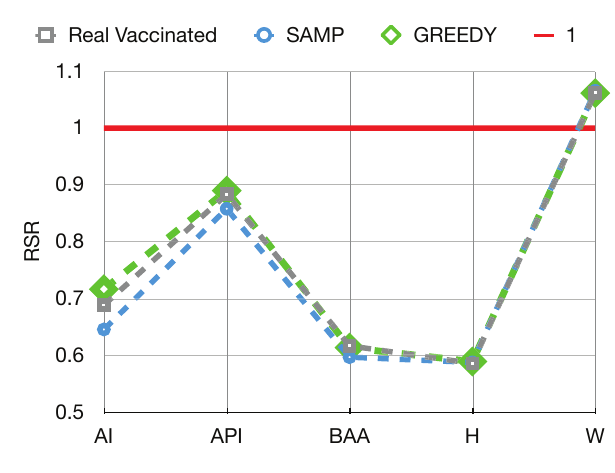}
    % \vspace{-2mm}
    \caption{RSR achieved when $\rho = 1$.}
    \label{fig:rsr_s1}
  \end{subfigure}
  \begin{subfigure}[b]{0.31\linewidth}
    \includegraphics[width=\linewidth]{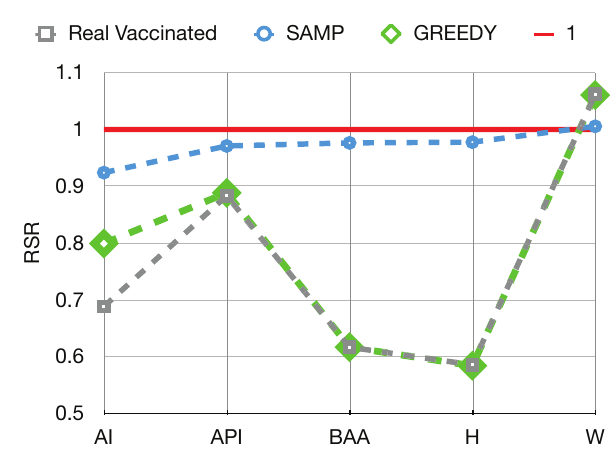}
    % \vspace{-2mm}
    \caption{RSR achieved when $\rho = 2$.}
    \label{fig:rsr_s2}
  \end{subfigure}
  \begin{subfigure}[b]{0.31\linewidth}
    \includegraphics[width=\linewidth]{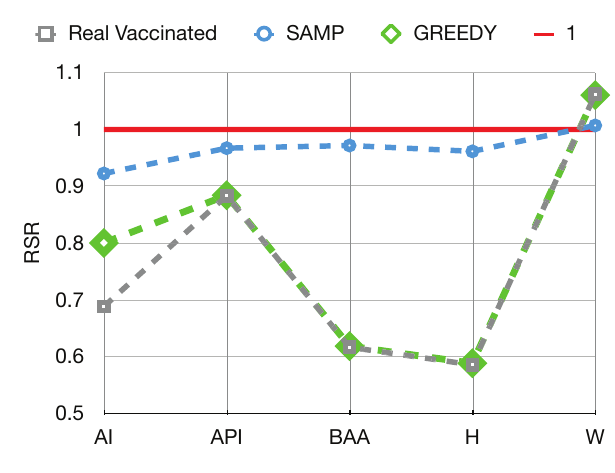}
    % \vspace{-2mm}
    \caption{RSR achieved when $\rho = 3$.}
    \label{fig:rsr_s3}
  \end{subfigure}
  \vspace{-2mm}
  \caption{RSR achieved for different races when the supply scarcity $\rho \in \{1,2,3\}$.}
  \label{fig:rsr}
\end{figure*} 
%From Figure~\ref{fig:people_vaccinated} we can observe that most of the doses were taken in some populous counties, e.g., Ramsey, Hennepin and Washinton. Thus we downsample the total serving capacity to focus on those populous counties and set $B$ to about $10,000$.
\xhdr{Algorithms}. In additional to the two LP-based algorithms \alga and \algb proposed in this paper, we also implement several heuristic baselines. Suppose a demand agent $j$ arrives at some time $t$ and let $\cN_{j,t}$ be the set of neighboring supply agents to $j$ that still have at least one remaining serving capacity at $t$.  If $\cN_{j,t}$ is empty, we have to reject $j$; otherwise,
 (a) \alggre: assign $j$ to the agent $i$ that has the largest serving capacity in $\cN_{j,t}$  (and break the ties arbitrarily); (b) \alguniform: select an agent uniformly at random in $\cN_{j,t}$; (c)  \algrank: select the agent from $\cN_{j,t}$ that has the largest order following a pre-selected random order on $I$. 
 
 % \bluee{It seems that the performance of \algrank and \alguniform are missing in the real datasets, perhaps we should mention them and say they perform similarly to \alggre and thus, we omit them?}
 
 % If $\cN_j$ is empty, reject $j$; otherwise, select the first available neighbor $i \in \cN_j$; (b) \alguniform: If $\cN_j$ is empty, reject $j$; otherwise, select an available neighbor uniformly at random;(c) \algrank: Fix a uniform random permutation of $\sigma$ at the start; assign each arriving agent to the adjacent available offline agent who is earliest in this order.

\xhdr{Results and discussions}.
{We test different settings when the supply scarcity $\rho$ takes values in $\{1, 1.5, \textbf{2}, 2.5, 3\}$ while the minimum serving capacity is fixed at $\ubar{b}=1$.} For each setting, we run all algorithms for $100$ times and take the average as the final performance. The competitive ratios are computed by comparing the averaged performance to the optimal value of \LP~\eqref{obj-1}. We also consider another setting when we fix $\rho=2$ while vary the minimum serving capacity $\ubar{b}\in \{\textbf{1}, 3, 5, 7, 9, 11\}$, respectively. Here, we scale up $\ubar{b}$ by simply filtering out those supply agents with a serving capacity less than $\ubar{b}$.

Figure~\ref{fig:cr_scar} shows that in a typical over-demanded settings ($\rho > 1$), \alga outperforms all other benchmarks. This highlights the global-optimization power of \alga when supply falls short of demand. Figure~\ref{fig:cr_b} shows that \alga always stay above its theoretical lower bound over different choices of $\ubar{b}$, which confirms our theoretical prediction as suggested by Theorem~\ref{thm:main-1}. To better demonstrate the advantage of \alga, we visualize two metrics of equality, \ie the absolute-serving ratio (ASR) and relative-serving ratio (RSR), in Figure~\ref{fig:asr} and Figure~\ref{fig:rsr}, respectively. Since the performances of \algrank and \alguniform are both similar to \alggre, we omit them in Figure~\ref{fig:asr} and Figure~\ref{fig:rsr} for presentation convenience.
%The reason can be explained as follows: from Figure~\ref{fig:opt_scar}, we can observe that when supply scarcity is less or equal to $1.5$, the performance bottleneck is situated in the degree of disproportionality for the least represented race, i.e., the Hispanic (it has a disproportionality of $0.586$). In this situation, \alga will also suffer from the performance bottleneck and have no motivation to push the ASR for other races. 

Figure~\ref{fig:asr} shows that in over-demanded settings ($\rho = 2$ and $\rho = 3$), \alga demonstrates its power in pushing all races' ASR toward the preset targets, but not for the case when supply roughly meets demand ($\rho=1$). This is because in the latter case ($\rho=1$), performances of both \alga and the clairvoyant (\OPT) are both bottlenecked by the degree of disproportionality of the least represented race, which is Hispanic with $\ubar{\kap}=0.586$ in our context. Thus, \alga have no motivation to push the ASR for any races.  The bottleneck can be seen from Figure~\ref{fig:opt_scar}: when $\rho \le 1.5$, the optimal value of \LP~\eqref{obj-1} stays at the value $\ubar{\kap}=0.586$; while as $\rho$ increases to $2$, the bottleneck of \OPT shifts from $\ubar{\kap}=0.583$ to $1/\rho \le 0.5$. In the latter case, \alga can leverage the solution of the benchmark LP~\eqref{obj-1} by serving the underrepresented arrivals as many as possible, while selecting only a fraction to serve for those overrepresented.  Figure~\ref{fig:rsr} shows that under the metric of RSR, \gre's performance aligns well with the real-world case. 
This is mainly due to the fact that heuristic-based strategies like Greedy have been widely adopted in practical situations.  Figure~\ref{fig:rsr} also suggests that though \alga is not designed to promote RSR, it can have an excellent performance on RSR (almost equal to $1$) in over-demanded settings. This further highlights the superiority of \alga.

\section{Conclusion and future directions}
This paper presents two LP-based strategies that can effectively navigate the distribution of the limited resources toward preset target ratios. Theoretical and numerical results highlight the power of our policies in reducing the existing biases in the arriving population, which are commonly observed in resource allocation when demands far outnumber supplies. Observe that the current two metrics (ASR and RSR) both model the gap between the serving ratio achieved the preset target as the ratio of the former to the latter. How about replacing the relative-ratio-based gap with the absolute-difference-based gap? 

% \clearpage
\section*{Acknowledgments}
Pan Xu was partially supported by NSF CRII Award IIS-1948157. The authors would like to thank the anonymous reviewers for their helpful feedback.

% \newpage
% {\small
\bibliography{vaccine}
% }

% \begin{comment}
%%%%%%%%start of appendix
\onecolumn

\appendix

\section{Proof of Lemma~\ref{lem:g}}
\begin{proof}
(P1) is due to the work~\cite{ma2021dynamic} (See the proof of Theorem 1 there). 
Next, we show (P2).  Note that when $s=1$, we have
 \begin{align*}
 g(1,b) &=\E[\min(\Pois(b),b)]/b=1-\Pr[\Pois(b)<b]+\sum_{k=1}^{b-1}\sfe^{-b} b^{k-1}/(k-1)!\\
 &=1-\Pr[\Pois(b)<b]+\Pr[\Pois(b) \le b-2]=1-\sfe^{-b}b^{b-1}/(b-1)!=1-\sfe^{-b}b^{b}/b! \ge 1-\frac{1}{\sqrt{2\pi b}}.
 \end{align*}
 The last inequality above is due to Stirling's formula. We can verify that $1-\sfe^{-b}b^{b}/b!$ gets minimized at $b=1$ with a value equal to $1-1/\sfe$. Thus, we claim that $g(1,b) \ge \max (1-1/\sfe,  1-\frac{1}{\sqrt{2\pi b}})$. 
Now we show (P3). Consider a given $0<s <1$,  we see that
\[
g(s,b)=\frac{\E[\min(\Pois(b/s),b)]}{b} \ge \Pr[\Pois(b/s) \ge b]=1- \Pr[\Pois(b/s) < b]\ge 1-\exp\Big(-\frac{(b/s) \cdot (1-s)^2}{2}\Big),
\]
where the last inequality above follows from the lower-tail bound of a Poisson random variable due to~\cite{pois-tail}. Thus, we claim that for any given $b \ge 1$, $g(s,b)$ will approach $1$ when $s \rightarrow 0^{+}$. 
\end{proof} 

\section{Proof of Main Theorem~\ref{thm:main-3}}

\begin{example}\label{exam}
Consider an instance as follows. There are $|I|=n$ supply agents each has a unit serving capacity $b_i=1$ for $i \in I$; there are $|J|=n+1$ demand agent types, which consists of $n$ ``rare'' types with $\lam_j=1/n$ for $j \in [n]:=\{1,2,\ldots,n\}$ and a single ``common'' type with $\lam_j=n-1$ for $j=0$. For each rare type $j \in [n]$, it can be served by a supply agent $i=j$ only; for the common type $j=0$, it can be served by all the $n$ supply agents.

Set $\cG=J$ with each $g=\{j\}$, \ie we consider homogeneous groups where each group consists of one single type. Set $\mu_j=\lam_j/\lam$ for each $j \in J$, where $\lam=\sum_{j \in J} \lam_j=n$. In this case, $\kap_j=\mu_j/(\lam \cdot \mu_j)=1$ for every $j \in J$ and $\ubar{\kap}=\min_{j\in J} \kap_j=1$, \ie each type (or group) of demand agents has a perfect representation in the arriving population and it matches exactly what we expect in the target serving ratio. Observe that our ASR-based objective \tbf{Max-A} $\max \min_{g \in \cG} \frac{\E[X_g]}{\E[A] \cdot \mu_g}$ is reduced to $\max \min_{j \in J} \frac{\E[X_j]}{\lam_j}$.
\end{example}

\begin{lemma}\label{lem:app-2}
The clairvoyant optimal has an expected performance of $\OPT \ge 1-o(1)$ on Example~\ref{exam}, where $o(1)$ is a vanishing term when $n$ approaches infinity.  
\end{lemma}
\begin{proof}
Consider such an offline algorithm \ALG as follows: every time after observing the all arriving demand agents, first try to serve as many arrivals of rare types as possible and then try to serve those of the common type. Let $A_r=\sum_{j \in [n]} A_j$ be the number of total arrivals of rare types and let $A_c$ be that of the common type. 
We see that the common type will be served with a ratio at least 
$\E[\min(A_c, n-A_r)]/(n-1)=1-o(1)$, where $A_r \sim \Pois(1)$ and $A_c=\Pois(n-1)$. As for each rare type $j \in [n]$, we see it will be served with a ratio at least $\E[\min(A_j,1)]/(1/n)=\Pr[A_j \ge 1]/(1/n)=1-o(1)$, where $A_j \sim \Pois(1/n)$. Thus, we claim that \ALG achieves an ASR at least $1-o(1)$ on Example~\ref{exam}. Since $\OPT \ge \ALG$, we establish our claim.
\end{proof}

\begin{lemma}\label{lem:app-1}
The optimal values of the two benchmark LPs, \LP~\eqref{obj-1} and \LP~\eqref{obj-2}, are both equal to $1$. 
\end{lemma}
\begin{proof}
Since we consider homogeneous groups, the two LPs are reduced to the same. Let us focus on  \LP~\eqref{obj-2} and let $s^*$ be the optimal value. By Lemma~\ref{lem:mu}, we have $s^* \le \ubar{\kap}=1$. 

Now we show $s^* \ge 1$. Consider such a solution as follows: for each rare type $j \in [n]$, set $x_{i=j,j}=\lam_j=1/n$; for the common type $j=0$, set $x_{ij}=1-1/n$ for each $i \in [n]$; $s=1$. We can verify that each $x_j=\sum_{i \in \cN_i} x_{ij}=\lam_j$; each $x_i=\sum_{j \in \cN_i} x_{ij}=1=b_i$; each $x_j=s \cdot \lam \cdot \mu_j=\lam_j$. Thus, we claim that it is feasible and therefore $s^* \ge s=1$. 
\end{proof}

Lemma~\ref{lem:app-1} suggests that the two algorithms \alga and \algb are reduced to the same on Example~\ref{exam}. Consider such an optimal solution to the two benchmark LPs~
\eqref{obj-1} and~\eqref{obj-2} as follows: for each rare type $j \in [n]$, $x_{i=j,j}=1/n$; for the common type $j=0$, $x_{ij}=1-1/n$ for each $i \in [n]$, and $s^*=1$. We can verify that \alga and \algb run essentially in the same way on Example~\ref{exam}, both of which will choose the neighbor $i=j$ with probability $1$ if some rare type $j \in [n]$ arrives and will sample a neighbor $i \in [n]$ uniformly if the common type $j=0$ arrives.

\begin{lemma}\label{lem:app-3}
Both \alga and \algb achieve an online-competitive ratio no more than $1-1/\sfe+o(1)$ on Example~\ref{exam}.
\end{lemma}
\begin{proof}
Focus on a given rare type $j \in [n]$. Let $X_j$ denote the number of agents of type $j$ that gets served in \alga (or \algb). For the supply agent $i=j$, we see that its serving capacity will be consumed with a rate equal to $x_i=\sum_{j\in \cN_i}x_{ij}=1$. Therefore,
\[
\frac{\E[X_j]}{\lam_j}=\frac{1}{\lam_j}\int_0^1 \Pr[\Pois(t)<1] \cdot \frac{1}{n} dt=1-\frac{1}{\sfe}.
\] 
Thus, we claim that the ASR achieved by \alga (or \algb) will be $\min_{j \in J} \E[X_j]/\lam_j \le 1-1/\sfe$. From Lemma~\ref{lem:app-2}, $\OPT \ge 1-o(1)$, thus we claim that both  \alga and \algb achieve an online-competitive ratio no more than $(1-1/\sfe)/(1-o(1))=1-1/\sfe+o(1)$.
\end{proof}

Observe that on Example~\ref{exam}, $\ubar{b}=\ubar{\kap}=s^*=1$. We can verify that $g(1,\ubar{b})=\ubar{\kap} \cdot g(s^*,\ubar{b})=g(1,1)=1-1/\sfe$. This suggests the tightness of online-competitive analyses of \alga and \algb.

\begin{lemma}\label{lem:app-4}
No algorithm can achieve an online-competitive ratio better than $\sqrt{3}-1$ on Example~\ref{exam}.
\end{lemma}

\begin{proof}
Note that any online algorithm which is going to reject the common type is better off doing so sooner rather than later, since an earlier rejection allows more time to observe which rare types arrive, and give those types priority.
For any $\tau\in[0,1]$, suppose that an algorithm, denoted by $\ALG(\tau)$, starts accepting the common type $j=0$ after time $\tau$. Recall from Example~\ref{exam} that any algorithm must have some (possibly randomized) order of supply agents to use when it wants to serve the common type. 

Let $\sig$ be any given random order on $[n]$ and $\ALG(\tau)$ will follow the order of $\sig$ to assign supply agents to serve the arriving common types after time $\tau$. For each $\ell \in [n]$, let $j=\sig(\ell)$ and $H_\ell$ be the arrival time of the $\ell$th common arrivals in a Poisson process with arrival rate $n-1$. Let $H=\min(\tau+H_\ell,1)$. We observe that the rare type of $j$ should arrive at least once during the time $ [0,H]$ such that it can be served by the demand agent $i=j$. Note that $H_\ell \sim \mathrm{Er}(\ell,n-1)$, an Erlang distribution with parameters $\ell$ and $n-1$, and $\E[H_\ell]=\ell/(n-1)$. Therefore, $\E[H] \le \min(\tau+\ell/(n-1),1)$. 
Let $X_j$ be the number of the rare type $j$ served and $H=\min(H_\ell,1)$
\[
\E[X_j]=\E_{H}\Big[1-\exp\Big(-H/n\Big) \Big] \le \E[H]/n \le \frac{1}{n} \min\Big(\tau+\ell/(n-1),1\Big).
\]
Thus, among all rare types, we have
\[
\min_{j \in [n]} \frac{\E[X_j]}{\lam_j}\le \frac{1}{n}\sum_{j=1}^n \frac{\E[X_j]}{\lam_j}=\sum_{j=1}^n \E[X_j] \le \frac{1}{n} \sum_{\ell=1}^n \min\Big(\tau+\ell/(n-1),1\Big)=
\int_0^1 \min(\tau+z,1) dz+o(1)=\Big(\tau+\frac{1}{2}-\frac{1}{2}\tau^2\Big)+o(1),
\] 
where $o(1)$ is a vanishing term when $n \rightarrow \infty$. Therefore, even using a randomized order, there must exist a rare type whose ASR will be at most $(\tau+\frac{1}{2}-\frac{1}{2}\tau^2+o(1))$. Meanwhile, for any $\tau \in [0,1]$, the expected number of common types served can be at most $(n-1)(1-\tau)$, which suggest that the ASR on the common type should be at most $1-\tau$. Thus, the ASR achieved by $\ALG(\tau)$ cannot exceed $\min \big(\tau+\frac{1}{2}-\frac{1}{2}\tau^2,1-\tau\big)$ when $n$ is sufficiently large. We can verify that the fairness is maximized at $\tau=2-\sqrt{3}$, in which case it equals $\sqrt{3}-1$. Note that the best clairvoyant algorithm have an expected performance at least $1-o(1)$. This completes the proof.
\end{proof}

\section{Experiments on homogeneous groups}
%We further test our proposed algorithms on a special case of homogeneous groups.

\xhdr{Setup of the synthetic dataset}.
We first generate $500$ supply agents and $500$ demand agents, respectively. The edge set $E$ is generated from a random graph with an average degree of $10$ for each demand agent. For each supply agent $i$, we set a uniform serving capacity $b = 5$, thus we have $\ubar{b} = 5$. For each demand agent $j$, we always adjust $\lambda_j$ such that $\sum_{j\in J}\lambda_j/\lambda = 1$. Then, we can get a target $\ubar{\kappa}$ by setting $\mu_j$ as $\mu_j = \frac{\lambda_j}{\lambda \cdot \kappa_j}$ for all $j \in J$, where $\kappa_j$ is a value (with up to one decimal) sampled uniformly at random from $[\ubar{\kappa},2-\ubar{\kappa}]$.
In addition, we focus on a practical case where resources are highly scarce with supply scarcity $\rho=2$.

\begin{figure}[ht!]
    \centering
    \includegraphics[scale=0.4]{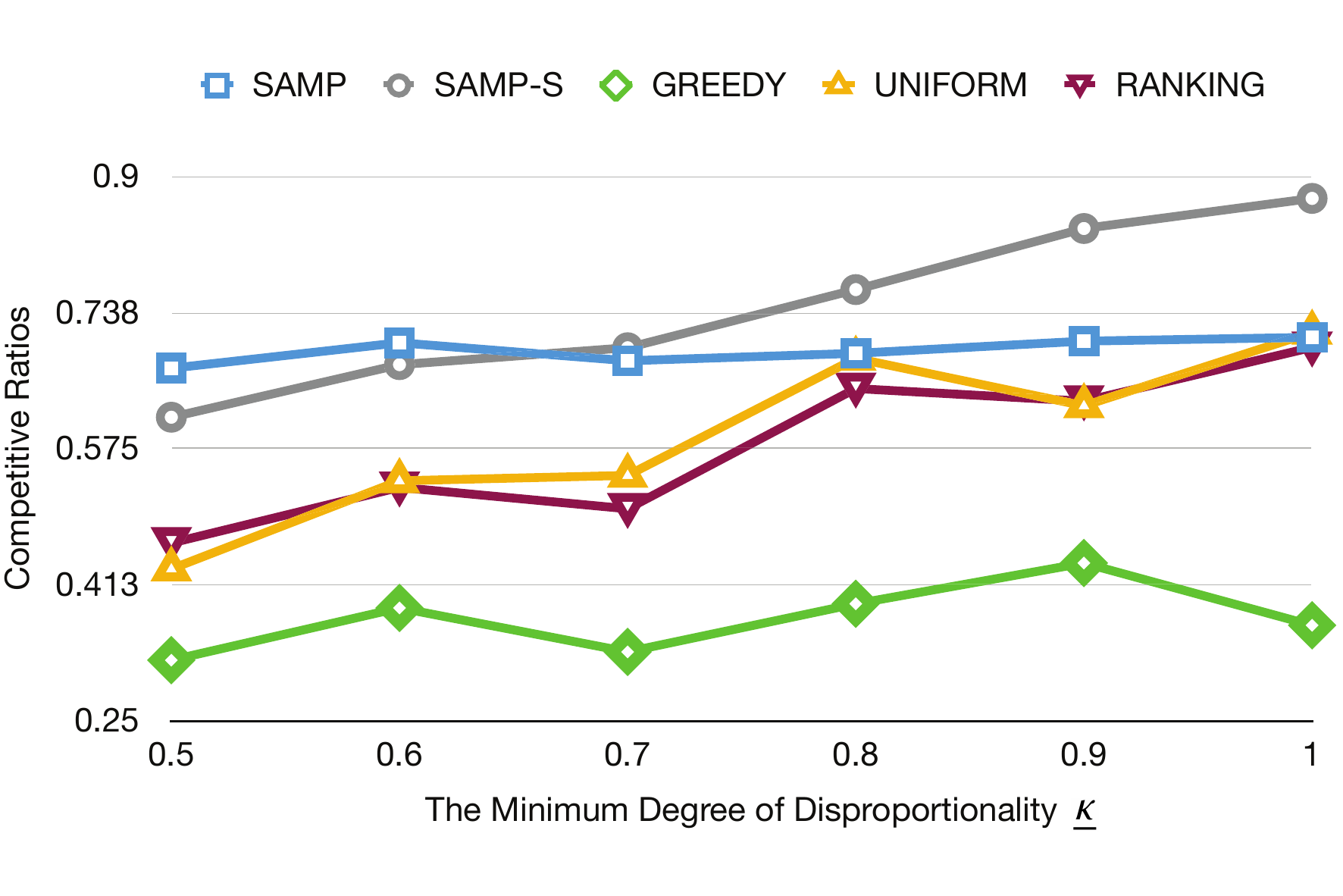}
    \caption{The synthetic dataset: competitive ratios achieved with $\ubar{\kap} \in \{0.5, 0.6, 0.7, 0.8, 0.9,1\}$.}
    \label{fig:syn_cr_kappa}
\end{figure}

\xhdr{Results and discussions}. In addition to the four algorithms, \ie  \alga, \alggre, \alguniform and \algrank, we also run \algb on the synthetic dataset, which is designed for the setting of homogeneous groups.
Overall, Figure~\ref{fig:syn_cr_kappa} shows that as $\ubar{\kappa}$ increases, the ratio of \algb will increase. When $\ubar{\kappa} \ge 0.6$, \algb starts to outperform \alga. This is consistent with results in Theorem~\ref{thm:main-2}. 

% \redd{Suggestions: (1) we can remove the line of theorecal lower bound; (2) add the symbol $\ubar{\kap}$ at the end of ``The Minimum Degree of Disproportionality" in  the figure. }

%This is partial due to Lemma~\ref{lem:g}: Observe that the supply scarcity $\rho=2$, and thus $s^*\le B/\lam=1/\rho=0.5$ by Lemma~\ref{lem:mu}, and we can verify that $g(s^*,\ubar{b}) \ge g(0.5, XX)$

%e explanation is that the supply scarcity is set as $2$, which implies a over-demanded setting, thus we have $s^* \ll 1$ and $g(s^*,\ubar{b}) \sim 1$ by Lemma~\ref{lem:g}. In this context, \algb will outperform \alga when the arriving population has a relatively low disproportionality overall.Note that \algb cannot always beat its lower bound suggested by Theorem~\ref{thm:main-2}, this is due to the fact that we can only imitate the Poisson arrival process for online agents in the simulation, which is a necessary component in establishing the theoretical bound.
% \end{comment}

	\end{document}